\documentclass[english]{article}
\usepackage{lmodern}

\usepackage[T1]{fontenc}
\usepackage[latin9]{inputenc}
\usepackage{geometry}
\usepackage{units}
\usepackage{dsfont}
\usepackage{amsmath}
\usepackage{amsthm}
\usepackage{amssymb}

\makeatletter
\numberwithin{figure}{section}
\numberwithin{equation}{section}
\theoremstyle{plain}
\newtheorem{thm}{\protect\theoremname}[section]
\theoremstyle{plain}
\newtheorem{cor}[thm]{\protect\corollaryname}
\theoremstyle{plain}
\newtheorem{lem}[thm]{\protect\lemmaname}
\theoremstyle{plain}
\newtheorem{prop}[thm]{\protect\propositionname}
\theoremstyle{remark}
\newtheorem*{acknowledgement*}{\protect\acknowledgementname}

\usepackage{float,psfrag,epsfig,color,url,hyperref}
\usepackage{algorithm,algorithmic}
\usepackage{graphicx,relsize}
\usepackage{amssymb,amsfonts,amsmath,amsthm,amscd,dsfont,mathrsfs,mathtools,microtype,nicefrac,pifont}
\usepackage{upgreek}
\usepackage[dvipsnames]{xcolor}
\usepackage{epstopdf,bbm,enumitem}
\usepackage{dsfont,tikz}
\usepackage[mathscr]{euscript}
\usepackage[toc,page]{appendix}
\hypersetup{
  colorlinks,
  linkcolor={red!50!black},
  citecolor={blue!50!black},
  urlcolor={blue!80!black}
}
\usepackage{etoolbox}
\patchcmd{\thebibliography}
  {\settowidth}
  {\setlength{\itemsep}{0pt plus 0.1pt}\settowidth}
  {}{}
\apptocmd{\thebibliography}
  {\small}
  {}{}
\allowdisplaybreaks
\usepackage{custom}
\makeatletter
\renewcommand{\paragraph}{%
  \@startsection{paragraph}{4}%
  {\z@}{1.25ex \@plus 1ex \@minus .2ex}{-1em}%
  {\normalfont\normalsize\bfseries}%
}
\makeatother

\makeatother

\usepackage{babel}
\providecommand{\acknowledgementname}{Acknowledgement}
\providecommand{\corollaryname}{Corollary}
\providecommand{\lemmaname}{Lemma}
\providecommand{\propositionname}{Proposition}
\providecommand{\theoremname}{Theorem}

\begin{document}
\def\balign#1\ealign{\begin{align}#1\end{align}}
\def\baligns#1\ealigns{\begin{align*}#1\end{align*}}
\def\balignat#1\ealign{\begin{alignat}#1\end{alignat}}
\def\balignats#1\ealigns{\begin{alignat*}#1\end{alignat*}}
\def\bitemize#1\eitemize{\begin{itemize}#1\end{itemize}}
\def\benumerate#1\eenumerate{\begin{enumerate}#1\end{enumerate}}

\newenvironment{talign*}
 {\let\displaystyle\textstyle\csname align*\endcsname}
 {\endalign}
\newenvironment{talign}
 {\let\displaystyle\textstyle\csname align\endcsname}
 {\endalign}

\def\balignst#1\ealignst{\begin{talign*}#1\end{talign*}}
\def\balignt#1\ealignt{\begin{talign}#1\end{talign}}

\let\originalleft\left
\let\originalright\right
\renewcommand{\left}{\mathopen{}\mathclose\bgroup\originalleft}
\renewcommand{\right}{\aftergroup\egroup\originalright}

\def\Gronwall{Gr\"onwall\xspace}
\def\Holder{H\"older\xspace}
\def\Ito{It\^o\xspace}
\def\Nystrom{Nystr\"om\xspace}
\def\Schatten{Sch\"atten\xspace}
\def\Matern{Mat\'ern\xspace}

\def\tinycitep*#1{{\tiny\citep*{#1}}}
\def\tinycitealt*#1{{\tiny\citealt*{#1}}}
\def\tinycite*#1{{\tiny\cite*{#1}}}
\def\smallcitep*#1{{\scriptsize\citep*{#1}}}
\def\smallcitealt*#1{{\scriptsize\citealt*{#1}}}
\def\smallcite*#1{{\scriptsize\cite*{#1}}}

\def\blue#1{\textcolor{blue}{{#1}}}
\def\green#1{\textcolor{green}{{#1}}}
\def\orange#1{\textcolor{orange}{{#1}}}
\def\purple#1{\textcolor{purple}{{#1}}}
\def\red#1{\textcolor{red}{{#1}}}
\def\teal#1{\textcolor{teal}{{#1}}}

\def\mbi#1{\boldsymbol{#1}} 
\def\mbf#1{\mathbf{#1}}
\def\mrm#1{\mathrm{#1}}
\def\tbf#1{\textbf{#1}}
\def\tsc#1{\textsc{#1}}

\def\mbiA{\mbi{A}}
\def\mbiB{\mbi{B}}
\def\mbiC{\mbi{C}}
\def\mbiDelta{\mbi{\Delta}}
\def\mbif{\mbi{f}}
\def\mbiF{\mbi{F}}
\def\mbih{\mbi{g}}
\def\mbiG{\mbi{G}}
\def\mbih{\mbi{h}}
\def\mbiH{\mbi{H}}
\def\mbiI{\mbi{I}}
\def\mbim{\mbi{m}}
\def\mbiP{\mbi{P}}
\def\mbiQ{\mbi{Q}}
\def\mbiR{\mbi{R}}
\def\mbiv{\mbi{v}}
\def\mbiV{\mbi{V}}
\def\mbiW{\mbi{W}}
\def\mbiX{\mbi{X}}
\def\mbiY{\mbi{Y}}
\def\mbiZ{\mbi{Z}}

\def\textsum{{\textstyle\sum}} 
\def\textprod{{\textstyle\prod}} 
\def\textbigcap{{\textstyle\bigcap}} 
\def\textbigcup{{\textstyle\bigcup}} 

\def\reals{\mathbb{R}} 
\def\integers{\mathbb{Z}} 
\def\rationals{\mathbb{Q}} 
\def\naturals{\mathbb{N}} 
\def\complex{\mathbb{C}} 

\def\what#1{\widehat{#1}}

\def\twovec#1#2{\left[\begin{array}{c}{#1} \\ {#2}\end{array}\right]}
\def\threevec#1#2#3{\left[\begin{array}{c}{#1} \\ {#2} \\ {#3} \end{array}\right]}
\def\nvec#1#2#3{\left[\begin{array}{c}{#1} \\ {#2} \\ \vdots \\ {#3}\end{array}\right]} 

\def\maxeig#1{\lambda_{\mathrm{max}}\left({#1}\right)}
\def\mineig#1{\lambda_{\mathrm{min}}\left({#1}\right)}

\def\Re{\operatorname{Re}} 
\def\indic#1{\mbb{I}\left[{#1}\right]} 
\def\logarg#1{\log\left({#1}\right)} 
\def\polylog{\operatorname{polylog}}
\def\maxarg#1{\max\left({#1}\right)} 
\def\minarg#1{\min\left({#1}\right)} 
\def\Earg#1{\E\left[{#1}\right]}
\def\Esub#1{\E_{#1}}
\def\Esubarg#1#2{\E_{#1}\left[{#2}\right]}
\def\bigO#1{\mathcal{O}\left(#1\right)} 
\def\littleO#1{o(#1)} 
\def\P{\mbb{P}} 
\def\Parg#1{\P\left({#1}\right)}
\def\Psubarg#1#2{\P_{#1}\left[{#2}\right]}
\def\Trarg#1{\Tr\left[{#1}\right]} 
\def\trarg#1{\tr\left[{#1}\right]} 
\def\Cov{\mrm{Cov}} 
\def\Covarg#1{\Cov\left[{#1}\right]}
\def\Covsubarg#1#2{\Cov_{#1}\left[{#2}\right]}
\def\Corr{\mrm{Corr}} 
\def\Corrarg#1{\Corr\left[{#1}\right]}
\def\Corrsubarg#1#2{\Corr_{#1}\left[{#2}\right]}
\newcommand{\info}[3][{}]{\mathbb{I}_{#1}\left({#2};{#3}\right)} 
\newcommand{\staticexp}[1]{\operatorname{exp}(#1)} 
\newcommand{\loglihood}[0]{\mathcal{L}} 


\providecommand{\arccos}{\mathop\mathrm{arccos}}
\providecommand{\dom}{\mathop\mathrm{dom}}
\providecommand{\diag}{\mathop\mathrm{diag}}
\providecommand{\tr}{\mathop\mathrm{tr}}
\providecommand{\card}{\mathop\mathrm{card}}
\providecommand{\sign}{\mathop\mathrm{sign}}
\providecommand{\conv}{\mathop\mathrm{conv}} 
\def\rank#1{\mathrm{rank}({#1})}
\def\supp#1{\mathrm{supp}({#1})}

\providecommand{\minimize}{\mathop\mathrm{minimize}}
\providecommand{\maximize}{\mathop\mathrm{maximize}}
\providecommand{\subjectto}{\mathop\mathrm{subject\;to}}

\def\openright#1#2{\left[{#1}, {#2}\right)}

\ifdefined\nonewproofenvironments\else
\ifdefined\ispres\else
 
\fi
\fi
\makeatletter
\@addtoreset{equation}{section}
\makeatother
\def\theequation{\thesection.\arabic{equation}}

\newcommand{\cmark}{\ding{51}}

\newcommand{\xmark}{\ding{55}}

\newcommand{\eq}[1]{\begin{align}#1\end{align}}
\newcommand{\eqn}[1]{\begin{align*}#1\end{align*}}
\renewcommand{\Pr}[1]{\mathbb{P}\left( #1 \right)}
\newcommand{\Ex}[1]{\mathbb{E}\left[#1\right]}

\newcommand{\matt}[1]{{\textcolor{Maroon}{[Matt: #1]}}}
\newcommand{\kook}[1]{{\textcolor{blue}{[Kook: #1]}}}
\definecolor{OliveGreen}{rgb}{0,0.6,0}
\newcommand{\sv}[1]{{\textcolor{OliveGreen}{[Santosh: #1]}}}

\global\long\def\on#1{\operatorname{#1}}%

\global\long\def\bw{\mathsf{Ball\ walk}}%
\global\long\def\sw{\mathsf{Speedy\ walk}}%
\global\long\def\gw{\mathsf{Gaussian\ walk}}%
\global\long\def\ps{\mathsf{Proximal\ sampler}}%
\global\long\def\dw{\mathsf{Dikin\ walk}}%

\global\long\def\chr{\mathsf{Coordinate\ Hit\text{-}and\text{-}Run}}%
\global\long\def\har{\mathsf{Hit\text{-}and\text{-}Run}}%
\global\long\def\gc{\mathsf{Gaussian\ cooling}}%
\global\long\def\ino{\mathsf{\mathsf{In\text{-}and\text{-}Out}}}%
\global\long\def\tgc{\mathsf{Tilted\ Gaussian\ cooling}}%
\global\long\def\PS{\mathsf{PS}}%
\global\long\def\psunif{\mathsf{PS}_{\textup{unif}}}%
\global\long\def\psexp{\mathsf{PS}_{\textup{exp}}}%
\global\long\def\psann{\mathsf{PS}_{\textup{ann}}}%
\global\long\def\psgauss{\mathsf{PS}_{\textup{Gauss}}}%
\global\long\def\eval{\mathsf{Eval}}%
\global\long\def\mem{\mathsf{Mem}}%

\global\long\def\O{O}%
\global\long\def\Otilde{\widetilde{O}}%
\global\long\def\Omtilde{\widetilde{\Omega}}%

\global\long\def\E{\mathbb{E}}%
\global\long\def\Z{\mathbb{Z}}%
\global\long\def\P{\mathbb{P}}%
\global\long\def\N{\mathbb{N}}%

\global\long\def\R{\mathbb{R}}%
\global\long\def\Rd{\mathbb{R}^{d}}%
\global\long\def\Rdd{\mathbb{R}^{d\times d}}%
\global\long\def\Rn{\mathbb{R}^{n}}%
\global\long\def\Rnn{\mathbb{R}^{n\times n}}%

\global\long\def\psd{\mathbb{S}_{+}^{d}}%
\global\long\def\pd{\mathbb{S}_{++}^{d}}%

\global\long\def\defeq{\stackrel{\mathrm{{\scriptscriptstyle def}}}{=}}%

\global\long\def\veps{\varepsilon}%
\global\long\def\lda{\lambda}%
\global\long\def\vphi{\varphi}%
\global\long\def\K{\mathcal{K}}%

\global\long\def\half{\frac{1}{2}}%
\global\long\def\nhalf{\nicefrac{1}{2}}%
\global\long\def\texthalf{{\textstyle \frac{1}{2}}}%
\global\long\def\ltwo{L^{2}}%

\global\long\def\ind{\mathds{1}}%
\global\long\def\op{\mathsf{op}}%
\global\long\def\ch{\mathsf{Ch}}%
\global\long\def\kls{\mathsf{KLS}}%
\global\long\def\ts{\mathsf{Ts}}%
\global\long\def\hs{\textup{HS}}%
\global\long\def\ls{\textup{LS}}%
\global\long\def\frob{\textrm{F}}%

\global\long\def\cpi{C_{\mathsf{PI}}}%
\global\long\def\clsi{C_{\mathsf{LSI}}}%
\global\long\def\cch{C_{\mathsf{Ch}}}%
\global\long\def\clch{C_{\mathsf{logCh}}}%
\global\long\def\cexp{C_{\mathsf{exp}}}%
\global\long\def\cgauss{C_{\mathsf{Gauss}}}%

\global\long\def\chooses#1#2{_{#1}C_{#2}}%

\global\long\def\vol{\on{vol}}%

\global\long\def\sym{\on{sym}}%

\global\long\def\law{\on{law}}%

\global\long\def\tr{\on{tr}}%

\global\long\def\diag{\on{diag}}%

\global\long\def\diam{\on{diam}}%

\global\long\def\poly{\on{poly}}%

\global\long\def\polylog{\on{polylog}}%

\global\long\def\Diag{\on{Diag}}%

\global\long\def\inter{\on{int}}%

\global\long\def\esssup{\on{ess\,sup}}%

\global\long\def\proj{\on{Proj}}%

\global\long\def\e{\mathrm{e}}%

\global\long\def\id{\mathrm{id}}%

\global\long\def\supp{\on{supp}}%

\global\long\def\spanning{\on{span}}%

\global\long\def\rows{\on{row}}%

\global\long\def\cols{\on{col}}%

\global\long\def\rank{\on{rank}}%

\global\long\def\T{\mathsf{T}}%

\global\long\def\bs#1{\boldsymbol{#1}}%

\global\long\def\eu#1{\EuScript{#1}}%

\global\long\def\mb#1{\mathbf{#1}}%

\global\long\def\mbb#1{\mathbb{#1}}%

\global\long\def\mc#1{\mathcal{#1}}%

\global\long\def\mf#1{\mathfrak{#1}}%

\global\long\def\ms#1{\mathscr{#1}}%

\global\long\def\mss#1{\mathsf{#1}}%

\global\long\def\msf#1{\mathsf{#1}}%

\global\long\def\textint{{\textstyle \int}}%
\global\long\def\Dd{\mathrm{D}}%
\global\long\def\D{\mathrm{d}}%
\global\long\def\grad{\nabla}%
 
\global\long\def\hess{\nabla^{2}}%
 
\global\long\def\lapl{\triangle}%
 
\global\long\def\deriv#1#2{\frac{\D#1}{\D#2}}%
 
\global\long\def\pderiv#1#2{\frac{\partial#1}{\partial#2}}%
 
\global\long\def\de{\partial}%
\global\long\def\lagrange{\mathcal{L}}%
\global\long\def\Div{\on{div}}%

\global\long\def\Gsn{\mathcal{N}}%
 
\global\long\def\BeP{\textnormal{BeP}}%
 
\global\long\def\Ber{\textnormal{Ber}}%
 
\global\long\def\Bern{\textnormal{Bern}}%
 
\global\long\def\Bet{\textnormal{Beta}}%
 
\global\long\def\Beta{\textnormal{Beta}}%
 
\global\long\def\Bin{\textnormal{Bin}}%
 
\global\long\def\BP{\textnormal{BP}}%
 
\global\long\def\Dir{\textnormal{Dir}}%
 
\global\long\def\DP{\textnormal{DP}}%
 
\global\long\def\Exp{\textnormal{Exp}}%
 
\global\long\def\Gam{\textnormal{Gamma}}%
 
\global\long\def\GEM{\textnormal{GEM}}%
 
\global\long\def\HypGeo{\textnormal{HypGeo}}%
 
\global\long\def\Mult{\textnormal{Mult}}%
 
\global\long\def\NegMult{\textnormal{NegMult}}%
 
\global\long\def\Poi{\textnormal{Poi}}%
 
\global\long\def\Pois{\textnormal{Pois}}%
 
\global\long\def\Unif{\textnormal{Unif}}%

\global\long\def\bpar#1{\bigl(#1\bigr)}%
\global\long\def\Bpar#1{\Bigl(#1\Bigr)}%

\global\long\def\abs#1{|#1|}%
\global\long\def\babs#1{\bigl|#1\bigr|}%
\global\long\def\Babs#1{\Bigl|#1\Bigr|}%

\global\long\def\snorm#1{\|#1\|}%
\global\long\def\bnorm#1{\bigl\Vert#1\bigr\Vert}%
\global\long\def\Bnorm#1{\Bigl\Vert#1\Bigr\Vert}%

\global\long\def\sbrack#1{[#1]}%
\global\long\def\bbrack#1{\bigl[#1\bigr]}%
\global\long\def\Bbrack#1{\Bigl[#1\Bigr]}%

\global\long\def\sbrace#1{\{#1\}}%
\global\long\def\bbrace#1{\bigl\{#1\bigr\}}%
\global\long\def\Bbrace#1{\Bigl\{#1\Bigr\}}%

\global\long\def\Abs#1{\left\lvert #1\right\rvert }%
\global\long\def\Par#1{\left(#1\right)}%
\global\long\def\Brack#1{\left[#1\right]}%
\global\long\def\Brace#1{\left\{  #1\right\}  }%

\global\long\def\inner#1{\langle#1\rangle}%
 
\global\long\def\binner#1#2{\left\langle {#1},{#2}\right\rangle }%

\global\long\def\norm#1{\lVert#1\rVert}%
\global\long\def\onenorm#1{\norm{#1}_{1}}%
\global\long\def\twonorm#1{\norm{#1}_{2}}%
\global\long\def\infnorm#1{\norm{#1}_{\infty}}%
\global\long\def\fronorm#1{\norm{#1}_{\text{F}}}%
\global\long\def\nucnorm#1{\norm{#1}_{*}}%
\global\long\def\staticnorm#1{\|#1\|}%
\global\long\def\statictwonorm#1{\staticnorm{#1}_{2}}%

\global\long\def\mmid{\mathbin{\|}}%

\global\long\def\otilde#1{\widetilde{O}(#1)}%
\global\long\def\wtilde{\widetilde{W}}%
\global\long\def\wt#1{\widetilde{#1}}%

\global\long\def\KL{\msf{KL}}%
\global\long\def\dtv{d_{\textrm{\textup{TV}}}}%
\global\long\def\FI{\msf{FI}}%
\global\long\def\tv{\msf{TV}}%
\global\long\def\TV{\msf{TV}}%

\global\long\def\cov{\on{cov}}%
\global\long\def\var{\on{Var}}%
\global\long\def\ent{\on{Ent}}%

\global\long\def\cred#1{\textcolor{red}{#1}}%
\global\long\def\cblue#1{\textcolor{blue}{#1}}%
\global\long\def\cgreen#1{\textcolor{green}{#1}}%
\global\long\def\ccyan#1{\textcolor{cyan}{#1}}%

\global\long\def\iff{\Leftrightarrow}%
 
\global\long\def\textfrac#1#2{{\textstyle \frac{#1}{#2}}}%

\title{A Unified Complexity Bound for Logconcave Sampling\date{}\author{Yunbum Kook\\ Georgia Tech\\  \texttt{yb.kook@gatech.edu} \and Santosh S. Vempala\\ Georgia Tech\\ \texttt{vempala@gatech.edu}}}
\maketitle
\begin{abstract}
We give a simple, unified, and nearly tight bound for sampling arbitrary
logconcave distributions from a warm start using the $\ino$ algorithm
along with exponential lifting. The main new ingredient in the analysis
is an improved bound on the Poincar\'e constant of a lifted distribution.
As a consequence, the resulting convergence rate is nearly tight for
both constrained settings (e.g., Gaussian restricted to a convex body)
and well-conditioned settings (e.g., strongly logconcave and smooth
densities). 
\end{abstract}

\section{Introduction}

Let $\pi^{X}\propto e^{-V}$ be a full-dimensional arbitrary logconcave
probability measure on $\Rd$, where $V:\Rd\to\R\cup\{\infty\}$ is
a convex function. Without loss of generality, we assume that $V$
is lower semi-continuous (see Preliminaries below). By scaling and
translating the domain, we can assume that:
\begin{equation}
B(0,1)\subset\msf L_{g}:=\sbrace{x\in\Rd:V(x)-\inf V\leq10d}\,.\tag{ground-set}\label{eq:intro-ground-set}
\end{equation}
When $\pi^{X}$ is uniform on a convex body $\K$, this just means
$B(0,1)\subset\K$ and for general logconcave densities, it roughly
says that its \emph{effective support}\footnote{The ground set $\msf L_{g}$ takes up most of $\pi^{X}$-measure \cite[Lemma 5.16]{LV07geometry}.}
contains a large ball. The main question for this paper is the following: 

\emph{Given access to an evaluation oracle for $V$ (i.e., access
to $V(x)$ for query $x$), how many queries do we need to generate
a sample whose law is close enough to $\pi^{X}$?}

We provide the following bound in terms of the R\'enyi divergence
of the output distribution (see \eqref{eq:renyi-definition}). 
\begin{thm}
[Complexity of zeroth-order logconcave sampling]\label{thm:INO-warm}
For a convex function $V:\Rd\to\R\cup\{\infty\}$ presented by an
evaluation oracle, let $\pi\propto e^{-V}$ be the logconcave distribution
over $\Rd$ with $B(0,1)\subset\msf L_{g}$ and $\Lambda=\norm{\cov\pi}_{\op}$.
Given $\veps>0$, an initial distribution $\pi_{0}$, and $q\geq2\vee\Omtilde(\log(d^{2}\Lambda\log\frac{1}{\veps}))$
such that $M_{q}=\norm{\frac{\D\pi_{0}}{\D\pi}}_{L^{q}(\pi)}\leq10$,
there exists an algorithm that returns $X^{*}$ satisfying $\eu R_{q}(\law X^{*}\mmid\pi)\leq\veps$,
using $\Otilde(qd^{2}\Lambda\log^{3}\frac{1}{\veps})$\emph{ }\textup{\emph{evaluation
queries in expectation}}\emph{. }If $M_{\infty}=\norm{\frac{\D\pi_{0}}{\D\pi}}_{L^{\infty}(\pi)}\leq10$
and $R^{2}:=\E_{\pi}[\norm{\cdot}^{2}]$, then there exists an algorithm
that returns $X^{*}$ satisfying $\eu R_{\infty}(\law X^{*}\mmid\pi)\leq\veps$,
using $\Otilde(d^{2}R\Lambda^{1/2}\polylog\frac{1}{\veps})$\emph{
}\textup{\emph{evaluation queries in expectation.}}
\end{thm}

This result allows us to recover $\kappa d$-complexity for zeroth-order
sampling in the well-conditioned setting (i.e., $0\prec\alpha I_{d}\preceq\hess V\preceq\beta I_{d}$
on $\Rd$) from a warm start \cite{CCSW22improved,ALP+24explicit,ALZ24entropy}
using the same algorithm and analysis and thus can be viewed as a
first step toward unified complexity results for sampling in the general
setting and the well-conditioned setting. Below we use $a\vee b$
to denote $\max\{a,b\}$.
\begin{cor}
[Complexity for well-conditioned distributions]\label{cor:well-conditioned}
Consider $\alpha$-strongly logconcave and $\beta$-log smooth $\pi\propto e^{-V}$
presented by an evaluation oracle. Given $\veps>0$, an initial distribution
$\pi_{0}$, and $q\geq2\vee\Omtilde(\log(\kappa d\log\frac{1}{\veps}))$
such that $M_{q}=\norm{\frac{\D\pi_{0}}{\D\pi}}_{L^{q}(\pi)}\leq10$,
there exists an algorithm that returns $X^{*}$ satisfying $\eu R_{q}(\law X^{*}\mmid\pi)\leq\veps$,
using $\Otilde(q\kappa d\log^{3}\frac{1}{\veps})$\emph{ }\textup{\emph{evaluation
queries in expectation, where $\kappa:=\beta/\alpha$ is the condition
number of $V$.}}
\end{cor}

Previous complexity guarantees \cite{KV25sampling,KV25faster,KV26zeroLC}
needed an additional $d^{2}$ query term. This is because in place
of $\Lambda$ in the corollary above, it was $\Lambda\vee1$; so in
the well-conditioned setting, the complexity guarantee ends up being
$\kappa d+d^{2}$. Where does this costly ``$\vee1$'' term come from?

The proximal sampler~\cite{LST21structured} is a well-studied algorithm
for logconcave sampling: for step size $h>0$, it alternates two steps:
(1) {[}Forward{]} $Y_{k+1}\sim\pi^{Y|X}(y\mid X_{k})=\msf N(X_{k},hI_{d})$,
and (2) {[}Backward{]} $X_{k+1}\sim\pi^{X|Y}(x\mid Y_{k+1})\propto\exp(-V(x)-\frac{1}{2h}\,\norm{x-Y_{k+1}}^{2})$.
It turns out that its convergence rate has a nice connection to the
\emph{Poincar\'e inequality}~\cite{CCSW22improved,KO25strong}.
A probability measure $\pi$ on $\Rd$ is said to satisfy a Poincar\'e
inequality\emph{ }with constant $C$ if for any locally Lipschitz
function $f\in L^{2}(\pi)$,
\begin{equation}
\var_{\pi}f:=\int\Bpar{f-\int f\,\D\pi}^{2}\,\D\pi\leq C\int\norm{\nabla f}^{2}\,\D\pi\,,\tag{\ensuremath{\msf{PI}}}\label{eq:pi}
\end{equation}
and the smallest such $C$ is called the Poincar\'e constant $\cpi(\pi)$.
The convergence rate of the proximal sampler in $\chi^{2}$-divergence
is simply $h^{-1}\cpi(\pi^{X})$.

Given access only to an evaluation oracle for convex $V$, how does
one implement the backward step? The $\ino$ algorithm~\cite{KVZ26INO,KV25sampling}
(a.k.a.\ the proximal sampler $\PS$) adapted the proximal sampler
to this setting via the \emph{exponential lifting}, which is inspired
by a conceptual connection between convex optimization (minimize $V$)
and logconcave sampling (sample from $\pi^{X}\propto e^{-V}$). For
$z:=(x,t)\in\Rd\times\R$, the exponential lifting is the probability
measure $\pi^{Z}$ on $\Rd\times\R$ with density 
\begin{equation}
\pi^{Z}(\D x\,\D t)\propto e^{-dt}\,\ind[V(x)\leq dt]\,\D x\,\D t\,.\tag{exp-lift}\label{eq:exp-lift}
\end{equation}
\begin{figure}
\includegraphics[width=\columnwidth]{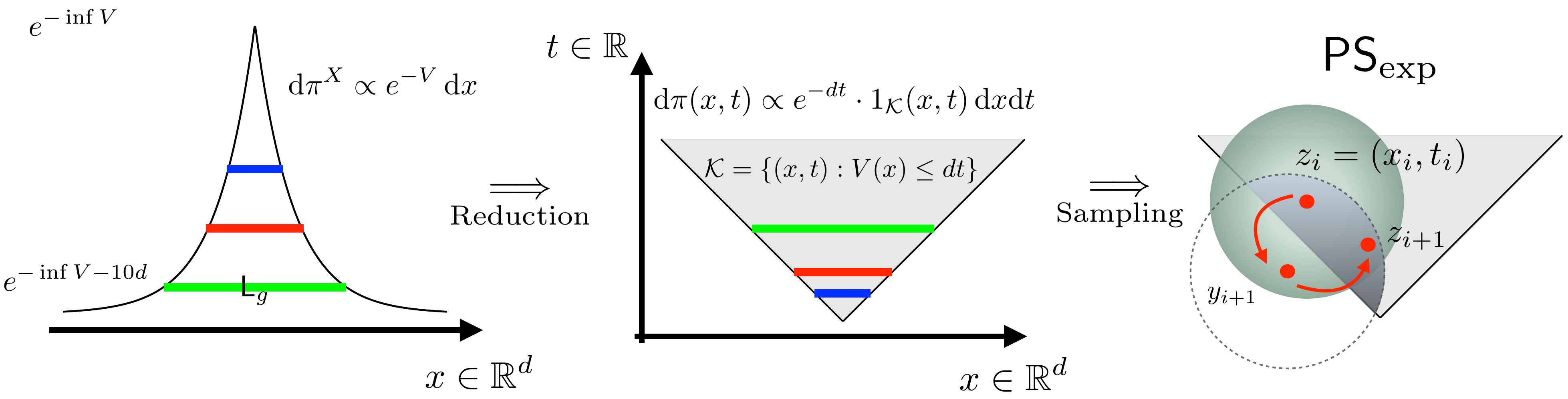} 

\caption{Exponential lifting and the proximal sampler ($\protect\ino$) adapted
from \cite{KV25sampling}.\label{fig:fig1}}
\end{figure}
Since its $X$-marginal is $\pi^{X}$, one may attempt to run $\ino$
to sample $(X,T)$ and then keep $X$ only. This lifting is useful
because its support is convex and its log-density is linear in the
lifted coordinate, thereby making the implementation of the backward
step straightforward.

The analysis now requires a bound on $\cpi(\pi^{Z})$ instead of $\cpi(\pi^{X})$.
Indeed, \cite{KV25sampling} showed that
\[
\norm{\cov\pi^{Z}}\lesssim\norm{\cov\pi^{X}}\vee1\,,
\]
where $\cov\pi$ is the covariance matrix with $\norm{\cov\pi}$ denoting
its operator norm. Applying the logarithmic KLS bound~\cite{Klartag23log}
(i.e., $\cpi(\pi)\lesssim\norm{\cov\pi}\log d$ for any logconcave
$\pi$ on $\Rd$), they established $\cpi(\pi^{Z})\lesssim_{\log}\cpi(\pi^{X})\vee1$. 

On the other hand, $\ino$ for uniform sampling over a convex body
has query complexity of $h^{-1}\norm{\cov\pi^{X}}$. this important
special case has an easy implementation of the backward step of $\ino$
not requiring the exponential lifting. This separation between the
query complexities for the special case and the general setting complicates
the design and analysis of downstream algorithms such as isotropic
rounding and integration.

Given this context, it is natural to wonder whether the ``$\vee1$''
term is really inevitable, or merely an artifact of the analysis?
In this paper, we show that the additive part ``$\vee1$'' can indeed
be removed, by improving the Poincar\'e constant of the lifted distribution.
This unifies the complexity of uniform sampling and general logconcave
sampling.
\begin{lem}
\label{lem:main} $\norm{\cov\pi^{Z}}\leq\norm{\cov\pi^{X}}+\frac{2}{d}$,
and $\frac{1}{d}\lesssim\norm{\cov\pi^{X}}$ under \eqref{eq:intro-ground-set}.
In particular, $\cpi(\pi^{Z})\lesssim\norm{\cov\pi^{X}}\log d\lesssim\cpi(\pi^{X})\log d$.
\end{lem}

This gives a unified mixing rate of $h^{-1}\norm{\cov\pi^{X}}$ (up
to a logarithmic factor) of $\ino$ for arbitrary logconcave sampling.
\cite{KV25sampling} showed that the backward step can be implemented
efficiently using rejection sampling by setting $h\approx d^{-2}$,
and this gives the guarantees of Theorem~\ref{thm:INO-warm}. The
improvement on $\cpi(\pi^{Z})$ unifies the query complexity of arbitrary
warm-start logconcave sampling, streamlining the guarantees established
in \cite{KV26zeroLC,KV25sampling}.

\paragraph{Technical ideas.}

By the Cauchy--Schwarz inequality, one can check that $\norm{\cov\pi^{Z}}\lesssim\norm{\cov\pi^{X}}+\var_{\pi^{Z}}T$.
Previously, \cite{KV25sampling} showed that $\var_{\pi^{Z}}T=O(1)$.
We improve this to $2/d$. We start with a simple but crucial observation
that if $X\sim\pi^{X}$ and $E\sim\Exp(1)$ are independent, then
$Z=(X,T)$ has the same distribution as $(X,(V(X)+E)/d)$:
\begin{equation}
Z\overset{d}{=}\bpar{X,\frac{V(X)+E}{d}}\,.\label{eq:lift-representation}
\end{equation}
Indeed, conditional on $X=x$, the variable $(V(X)+E)/d$ has density
proportional to $e^{-dt}\,\ind[t\geq\frac{V(x)}{d}]$. Hence, $T=(V(X)+E)/d$
in distribution, and 
\[
\var_{\pi^{Z}}T=\frac{\var_{\pi^{X}}V(X)+1}{d^{2}}=\frac{\var_{\pi^{X}}(\log\pi^{X})+1}{d^{2}}\,.
\]

This observation prompts a crisp connection to \emph{varentropy}.
The varentropy is the variance of the information content, namely
$-\log\pi(X)$ of a random point $X$ with density $\pi$. Bobkov
and Madiman proved an $O(d)$-bound for this quantity for logconcave
densities \cite{BM11concentration}. The sharp constant was later
identified independently in the theses of Nguyen and Wang \cite{Nguyen13thesis,Wang14thesis}:
\begin{equation}
\var_{\pi^{X}}V(X)\leq d\,.\tag{varentropy}\label{eq:varentropy}
\end{equation}
We refer the readers to the work of Fradelizi, Madiman, and Wang \cite[Theorem~2.3]{FMW16Optimal}
for a proof. This result plays an important role in establishing $d$-self-concordance
of the entropic barrier in the work of Chewi \cite{Chewi23EntropicBarrier}.
In our setting, $-\log\pi(X)=V(X)+O(1)$. Therefore, $\norm{\cov\pi^{Z}}\leq\norm{\cov\pi^{X}}+\frac{2}{d}$.
Moreover, $1/d$ is unavoidable under the ground-set normalization
\eqref{eq:intro-ground-set}. In fact, one can show that $1/d\lesssim\norm{\cov\pi^{X}}(\leq\cpi(\pi^{X}))$
in this setting. 

\paragraph{Preliminaries.}

For probability measures $\mu,\nu$ with $\mu\ll\nu$ and $q\in(1,\infty)$,
the \emph{$q$-R\'enyi divergence} is 
\begin{equation}
\eu R_{q}(\mu\mmid\nu):=\frac{1}{q-1}\,\log\int\bpar{\frac{\D\mu}{\D\nu}}^{q}\,\D\nu=\frac{q}{q-1}\,\log\,\Bnorm{\frac{\D\mu}{\D\nu}}_{L^{q}(\nu)}\,.\label{eq:renyi-definition}
\end{equation}
If $\mu$ is not absolutely continuous with respect to $\nu$, we
set $\eu R_{q}(\mu\mmid\nu)=\infty$. The R\'enyi-infinity divergence
is $\eu R_{\infty}(\mu\mmid\nu):=\log\esssup_{\nu}\frac{\D\mu}{\D\nu}$. 

We always use a lower semicontinuous convex representative of $V$.
Replacing a convex function by its lower semicontinuous closure does
not change $e^{-V}$ outside a Lebesgue-null set. The representative
matters only for pointwise statements such as the ground-set inclusion
\eqref{eq:intro-ground-set}; throughout the paper that inclusion
is an assumption on the chosen representative. In convex-analysis
terminology, this is the closed convex representative of $V$ \cite[\S7 and \S10]{Rockafellar70Convex}.

\section{Improved covariance of the lifted distribution via varentropy\label{sec:covariance}}

In this section, we prove Lemma~\ref{lem:main}. Let $\pi^{X}\propto e^{-V}$
be a full-dimensional logconcave probability measure on $\Rd$, and
$\pi^{Z}$ the lifted distribution \eqref{eq:exp-lift}. Note that
one can generate an $O(1)$-warm start for $\pi^{Z}$ from an $M$-warm
start for $\pi^{X}$ (see the proof of \cite[Theorem 2.15]{KV25sampling}).

\paragraph{First claim.}

We have already observed that $\var_{\pi^{Z}}T=\frac{\var_{\pi^{X}}(\log\pi^{X})+1}{d^{2}}$.
Using \eqref{eq:varentropy}, we establish $\var T\leq2/d$. It remains
to pass from this to the full covariance matrix. For every $(u,s)\in\Rd\times\R$
with $\norm u^{2}+s^{2}=1$, by the Cauchy--Schwarz inequality
\begin{align*}
\var_{\pi^{Z}}(u^{\T}X+sT) & =u^{\T}\cov\pi^{X}u+2\cov_{\pi^{Z}}(u^{\T}X,sT)+s^{2}\var_{\pi^{Z}}T\le\bpar{\norm{\cov\pi^{X}}^{1/2}\,\norm u+\sqrt{\var_{\pi^{Z}}T}\,\abs s}^{2}\\
 & \le(\norm{\cov\pi^{X}}+\var_{\pi^{Z}}T)(\norm u^{2}+s^{2})\leq\norm{\cov\pi^{X}}+\frac{2}{d}\,.
\end{align*}
Taking the supremum over unit $(u,s)$ proves the first claim.

\paragraph{Second claim.}

The next lemma shows that $O(1/d)$ is absorbed by the intrinsic covariance
scale under our ground-set normalization \eqref{eq:intro-ground-set}.
\begin{lem}
\label{lem:ground-set-lower-scale} Assume the lower-semicontinuous
convex representative of $V$ satisfies \eqref{eq:intro-ground-set}.
Then 
\begin{equation}
\tr(\cov\pi^{X})\gtrsim1\,,\qquad\cpi(\pi^{X})\geq\norm{\cov\pi^{X}}\gtrsim\frac{1}{d}\,.\label{eq:ground-set-scale-main}
\end{equation}
\end{lem}

\begin{proof}
Let $f=f_{X}=e^{-V}/Z_{X}$. Since the logconcave density $f$ has
a mode, it should attain a maximum. Indeed, any level set of $V$
should be bounded (otherwise, the density is not integrable). By \cite[Theorem 1.9]{RW98variational},
$m:=\inf V$ is attained.

Let 
\[
M:=\norm f_{\infty}=\frac{e^{-m}}{Z_{X}}\,.
\]
The ground-set assumption gives $f(x)\geq e^{-10d}M$ on $B(0,1)$.
Thus 
\[
1=\int f(x)\,\D x\geq e^{-10d}M\,\vol B_{2}^{d}\,,
\]
and hence 
\begin{equation}
M\leq\frac{e^{10d}}{\vol(B_{2}^{d})}\leq(C\sqrt{d})^{d}\,.\label{eq:max-density-upper-main}
\end{equation}

Let $\mu=\E X$ and $\Sigma_{X}=\cov\pi^{X}$. Since $\pi^{X}$ is
full-dimensional and logconcave, $\Sigma_{X}$ is positive definite.
The isotropic random vector $Y=\Sigma_{X}^{-1/2}\,(X-\mu)$ has density
$g(y)\,\D y$, where 
\[
g(y)=\sqrt{\det\Sigma_{X}}\,f(\mu+\Sigma_{X}^{1/2}y)\,.
\]
Thus $\norm g_{\infty}=\sqrt{\det\Sigma_{X}}\,M$. The Gaussian maximum-entropy
bound gives 
\[
h(Y):=\E[-\log g(Y)]\leq\frac{d}{2}\log(2\pi e)\,,
\]
while $h(Y)=\E[-\log g(Y)]\geq-\log\norm g_{\infty}$. Therefore $\norm g_{\infty}\geq(2\pi e)^{-d/2}=c^{d}$,
and 
\[
\sqrt{\det\Sigma_{X}}\,M\geq c^{d}\,.
\]
Combining this with \eqref{eq:max-density-upper-main} yields 
\[
(\det\Sigma_{X})^{1/d}\gtrsim\frac{1}{d}\,.
\]
The arithmetic--geometric mean inequality gives $\tr\Sigma_{X}\geq d\,(\det\Sigma_{X})^{1/d}\gtrsim1$,
and also $\norm{\Sigma_{X}}\geq(\det\Sigma_{X})^{1/d}\gtrsim1/d$.
Finally, $\cpi(\pi^{X})\geq\norm{\cov\pi^{X}}$ by the Poincar\'e
inequality on linear functions. 
\end{proof}

\section{Unified complexity for logconcave sampling\label{sec:consequences}}

The improved covariance estimate has two consequences for sampling
from a warm start (see Theorem~\ref{thm:INO-warm}). For finite R\'enyi
guarantees, the relevant functional inequality is the Poincar\'e
inequality, and the covariance improvement directly affects the query
complexity. For an $\eu R_{\infty}$-guarantee, the natural mixing
statement is based on \emph{log-Sobolev inequality} (LSI): A probability
measure $\pi$ on $\Rd$ is said to satisfy a logarithmic Sobolev
inequality with constant $C$ if for any locally Lipschitz function
$f\in L^{2}(\pi)$,
\begin{equation}
\ent_{\pi}(f^{2}):=\int f^{2}\log f^{2}\,\D\pi-\int f^{2}\,\D\pi\cdot\log\int f^{2}\,\D\pi\leq2C\int\norm{\nabla f}^{2}\,\D\pi\,,\tag{\ensuremath{\msf{LSI}}}\label{eq:lsi}
\end{equation}
and the smallest such $C$ is referred to as the log-Sobolev constant
$\clsi(\pi)$. Starting from an $O(1)$-warm initial distribution,
the proximal sampler (or INO) mixes at a rate governed by the LSI
constant of the lifted target. 

\paragraph{$\protect\eu R_{q}$-guarantee from $\protect\eu R_{q}$-warmness
(PI).}

For a finite R\'enyi guarantee, the cleanest guarantee is the balanced
warm-start theorem of \cite{KV26zeroLC}, rather than the older $\eu R_{\infty}\to\eu R_{q}$
warm-start theorem in \cite{KV25sampling}. That theorem runs the
same exponential-lift sampler with proper restart, but requires a
weaker $L^{q}$-warm start for $\eu R_{q}$-guarantees.

Theorem~1.7 of \cite{KV26zeroLC} provides query complexity $\Otilde(qd^{2}(\Lambda\vee1)\log^{3}\frac{1}{\veps})$
from an $O(1)$ warm start in $L^{q}(\pi^{X})$. In that proof, the
term $\Lambda\vee1$ enters through the lifted Poincar\'e constant
$\cpi(\pi^{Z})\lesssim_{\log}\Lambda\vee1$ of the exponential lift.
Lemma~\ref{lem:main} improves it to $\cpi(\pi^{Z})\lesssim\Lambda\log d$.
The restart, rejection-sampling, and $\eu R_{q}$ contraction arguments
of \cite{KV26zeroLC} remain the same. Substituting this Poincar\'e
constant proves the first part of Theorem~\ref{thm:INO-warm}.

\paragraph{$\protect\eu R_{\infty}$-guarantee from $\protect\eu R_{\infty}$-warmness
(LSI).}

The finite-$q$ mixing result above is driven by a Poincar\'e constant,
while a pointwise $\eu R_{\infty}$ guarantee is stronger and is better
viewed through LSI. To ensure finite $\clsi$, we should truncate
the lifted convex set. When $R^{2}:=\E_{\pi^{X}}[\norm{\cdot}^{2}]$,
\cite{KV25sampling} showed the following result:
\begin{prop}
\label{prop:truncated-lsi} Under \eqref{eq:intro-ground-set}, given
$\veps\in(0,1)$, there exists a truncated convex body $K_{\tr}^{Z}\subset\Rd\times\R$
of diameter $R_{\veps}:=O((R\vee1)\,\log\frac{1}{\veps})$ such that
$\pi^{Z}(K_{\tr}^{Z})\geq1-\veps/2$ and $\pi_{\mathrm{tr}}^{Z}\propto\pi^{Z}\cdot\ind_{K_{\mathrm{tr}}^{Z}}$
satisfies 
\[
\clsi(\pi_{\mathrm{tr}}^{Z})\lesssim R_{\veps}\,\norm{\cov\pi^{X}}^{1/2}\log^{1/2}d\,.
\]
\end{prop}

\begin{proof}
For unit vector $w\in\Rd\times\R$, we have $\var_{\pi_{\mathrm{tr}}^{Z}}(w^{\T}Z)\leq\E_{\pi_{\mathrm{tr}}^{Z}}[(w^{\T}Z-\E_{\pi^{Z}}w^{\T}Z)^{2}]\leq2\var_{\pi^{Z}}(w^{\T}Z)$.
Thus, 
\[
\norm{\cov\pi_{\mathrm{tr}}^{Z}}\leq2\norm{\cov\pi^{Z}}\underset{\text{Lemma }\ref{lem:main}}{\lesssim}2\norm{\cov\pi^{X}}\,.
\]
Using the LSI bound for logconcave measure with bounded support \cite[Theorem~3.1]{KV25faster},
\[
\clsi(\pi_{\mathrm{tr}}^{Z})\lesssim R_{\veps}\cpi(\pi_{\mathrm{tr}}^{Z})^{1/2}\lesssim R_{\veps}\,\norm{\cov\pi_{\mathrm{tr}}^{Z}}^{1/2}\log^{1/2}d\lesssim R_{\veps}\,\norm{\cov\pi^{X}}^{1/2}\log^{1/2}d\,,
\]
which completes the proof.
\end{proof}
We are now ready to prove the second part of Theorem~\ref{thm:INO-warm}.
Let $\K:=K_{\tr}^{Z}$. Given the $M_{\infty}$-warm $\pi_{0}^{Z}$,
one can draw a sample from $\pi_{0}^{Z}|_{\K}\propto\pi_{0}^{Z}\cdot\ind_{\K}$
by rejection sampling: draw $Z^{*}\sim\pi_{0}^{Z}$ until $Z^{*}\in\K$.
Since $\pi^{Z}(\K)\geq1-\veps/2$, one can readily check that
\[
\frac{\pi_{0}^{Z}|_{\K}}{\pi^{Z}|_{\K}}\leq M_{\infty}\,\frac{1}{1-M_{\infty}\veps/2}\,.
\]
Since $M_{\infty}\leq10$, for $\veps<0.01$, the warmness of $\pi_{0}^{Z}|_{\K}$
with respect to $\pi^{Z}|_{\K}$ is also $O(1)$. Next, \cite[Lemma~3.6]{KV25sampling}
showed that the convergence rate of $\ino$ for $\eu R_{\infty}$-guarantee
is $h^{-1}\clsi(\pi_{\mathrm{tr}}^{Z})\log\frac{1}{\veps}$ from an
$O(1)$-warm start in $\eu R_{\infty}$. Using Proposition~\ref{prop:truncated-lsi}
and $R^{2}\geq\tr(\cov\pi^{X})\geq1$, from an $O(1)$-warm start
in $\eu R_{\infty}$, the query complexity for $\eu R_{\infty}$-guarantee
is simply
\[
d^{2}\clsi(\pi_{\mathrm{tr}}^{Z})\polylog\frac{1}{\veps}\lesssim d^{2}R_{\veps}\,\norm{\cov\pi^{X}}^{1/2}\log^{1/2}d\polylog\frac{1}{\veps}\lesssim d^{2}R\,\norm{\cov\pi^{X}}^{1/2}\log^{1/2}d\polylog\frac{1}{\veps}\,.
\]

\paragraph{Well-conditioned setting.}

We now prove Corollary~\ref{cor:well-conditioned}. Let $\D\pi\propto\exp(-V)\,\D x$
be well-conditioned, and $x^{*}:=\arg\min_{x\in\Rd}V(x)$. Since $V(x)\leq V(x^{*})+\frac{\beta}{2}\,\norm{x-x^{*}}^{2}$,
we clearly have that 
\[
B\bpar{x^{*},\sqrt{\nicefrac{20d}{\beta}}}=\bbrace{x\in\Rd:\frac{\beta}{2}\,\norm{x-x^{*}}^{2}\leq10d}\subseteq\msf L_{g}\,,
\]
which means the ground set includes a ball of radius $\Theta(\sqrt{d/\beta})$.
Also, since $\norm{\cov\pi}\leq\cpi(\pi)\leq\alpha^{-1}$ (e.g., by
the Brascamp--Lieb inequality), the scaling $S(x):=\sqrt{\beta/d}(x-x^{*})$
(equivalently, taking $h\approx(\beta d)^{-1}$) ensures $\cpi(S_{\#}\pi)\lesssim\kappa/d$
and satisfies the ground set condition \eqref{eq:intro-ground-set}.
Thus, from an $O(1)$-warm start, the query complexity becomes $qd^{2}\frac{\kappa}{d}=q\kappa d$
as claimed.
\begin{acknowledgement*}
This work was supported in part by NSF award AF-2504994 and a Simons
Investigator award. We acknowledge the use of ChatGPT for a pointer
to relevant literature, in particular \cite[Theorem~2.3]{FMW16Optimal},
and for help with proofreading an earlier draft.
\end{acknowledgement*}
\bibliographystyle{alpha}
\bibliography{main}

@book{RW98variational,
	author = {Rockafellar, R. Tyrrell and Wets, Roger J.-B.},
	publisher = {Springer-Verlag, Berlin},
	series = {Grundlehren der mathematischen Wissenschaften},
	title = {Variational analysis},
	volume = {317},
	year = {1998}}

@article{LV07geometry,
	author = {Lov\'{a}sz, L\'{a}szl\'{o} and Vempala, Santosh S.},
	journal = {Random Structures \& Algorithms},
	number = {3},
	pages = {307--358},
	title = {The geometry of logconcave functions and sampling algorithms},
	volume = {30},
	year = {2007}}

@inbook{Chewi23EntropicBarrier,
	address = {Cham},
	author = {Chewi, Sinho},
	booktitle = {Geometric Aspects of Functional Analysis: Israel Seminar (GAFA) 2020-2022},
	editor = {Eldan, Ronen and Klartag, Bo'az and Litvak, Alexander and Milman, Emanuel},
	pages = {209--222},
	publisher = {Springer International Publishing},
	title = {The entropic barrier is $n$-self-concordant},
	year = {2023}}

@article{ALP+24explicit,
	author = {Andrieu, Christophe and Lee, Anthony and Power, Sam and Wang, Andi Q.},
	journal = {The Annals of Applied Probability},
	number = {4},
	pages = {4022--4071},
	title = {Explicit convergence bounds for {M}etropolis {M}arkov chains: isoperimetry, spectral gaps and profiles},
	volume = {34},
	year = {2024}}

@article{ALZ24entropy,
	author = {Filippo Ascolani and Hugo Lavenant and Giacomo Zanella},
	journal = {arXiv preprint arXiv:2410.00858},
	title = {Entropy contraction of the {G}ibbs sampler under log-concavity},
	year = {2024}}

@inproceedings{LST21structured,
	author = {Lee, Yin Tat and Shen, Ruoqi and Tian, Kevin},
	booktitle = {Conference on Learning Theory},
	pages = {2993--3050},
	publisher = {PMLR},
	title = {Structured logconcave sampling with a restricted {G}aussian oracle},
	volume = {134},
	year = {2021}}

@article{KO25strong,
	author = {Klartag, Bo'az and Ordentlich, Or},
	journal = {IEEE Transactions on Information Theory},
	number = {5},
	pages = {3317--3333},
	title = {The strong data processing inequality under the heat flow},
	volume = {71},
	year = {2025}}

@inproceedings{CCSW22improved,
	author = {Chen, Yongxin and Chewi, Sinho and Salim, Adil and Wibisono, Andre},
	booktitle = {Conference on Learning Theory},
	pages = {2984--3014},
	publisher = {PMLR},
	title = {Improved analysis for a proximal algorithm for sampling},
	volume = {178},
	year = {2022}}

@article{KVZ26INO,
	author = {Kook, Yunbum and Vempala, Santosh S. and Zhang, Matthew S.},
	journal = {Random Structures \& Algorithms},
	number = {3},
	pages = {e70061},
	title = {In-and-{O}ut: algorithmic diffusion for sampling convex bodies},
	volume = {68},
	year = {2026}}

@article{BM11concentration,
	author = {Bobkov, Sergey and Madiman, Mokshay},
	journal = {The Annals of Probability},
	number = {4},
	pages = {1528--1543},
	title = {Concentration of the information in data with log-concave distributions},
	volume = {39},
	year = {2011}}

@inbook{FMW16Optimal,
	author = {Fradelizi, Matthieu and Madiman, Mokshay and Wang, Liyao},
	pages = {45--60},
	publisher = {Springer International Publishing},
	title = {Optimal Concentration of Information Content for Log-Concave Densities},
	year = {2016}}

@article{Klartag23log,
	author = {Klartag, Bo'az},
	journal = {Ars Inveniendi Analytica},
	number = {4},
	pages = {1--17},
	title = {Logarithmic bounds for isoperimetry and slices of convex sets},
	year = {2023}}

@phdthesis{Nguyen13thesis,
	author = {Nguyen, Van Hoang},
	month = {October},
	school = {Universit{\'e} Pierre et Marie Curie (Paris VI)},
	title = {In{\'e}galit{\'e}s fonctionnelles et convexit{\'e}},
	year = {2013}}

@phdthesis{Wang14thesis,
	author = {Wang, Liyao},
	month = {May},
	school = {Yale University},
	title = {Heat capacity bound, energy fluctuations and convexity},
	year = {2014}}

@book{Rockafellar70Convex,
	author = {Rockafellar, R. Tyrrell},
	publisher = {Princeton University Press, Princeton, NJ},
	series = {Princeton Mathematical Series, No. 28},
	title = {Convex analysis},
	year = {1970}}

@inproceedings{KV25faster,
	author = {Kook, Yunbum and Vempala, Santosh S.},
	booktitle = {Symposium on Foundations of Computer Science},
	pages = {997-1006},
	publisher = {IEEE},
	title = {Faster logconcave sampling from a cold start in high dimension},
	year = {2025}}

@inproceedings{KV25sampling,
	author = {Kook, Yunbum and Vempala, Santosh S.},
	booktitle = {{S}ymposium on {T}heory of {C}omputing},
	pages = {924--932},
	publisher = {ACM},
	title = {Sampling and integration of logconcave functions by algorithmic diffusion},
	year = {2025}}

@article{KV26zeroLC,
	author = {Kook, Yunbum and Vempala, Santosh S.},
	journal = {arXiv preprint arXiv:2507.18021},
	title = {Zeroth-order logconcave sampling},
	year = {2025}}

\end{document}